%% file: dt_incr_syn.tex
\documentclass[letterpaper, 10 pt, conference]{ieeeconf}  

\IEEEoverridecommandlockouts                              

\usepackage{graphics} 
\usepackage{amsmath,amssymb,amsthm} 
\usepackage{mathrsfs}
\usepackage[noadjust]{cite}
\let\labelindent\relax\usepackage{enumitem}

\newtheorem{thm}{Theorem}
\newtheorem{lem}[thm]{Lemma}

\theoremstyle{definition}
\newtheorem{defn}[thm]{Definition}

\theoremstyle{remark}
\newtheorem{rem}[thm]{Remark}

\input{CustomCommands}

\DeclareMathAlphabet{\mathsfit}{T1}{\sfdefault}{\mddefault}{\sldefault}
\SetMathAlphabet{\mathsfit}{bold}{T1}{\sfdefault}{\bfdefault}{\sldefault}

\makeatletter
\g@addto@macro\normalsize{%
  \setlength\abovedisplayskip{.3em}
  \setlength\belowdisplayskip{\abovedisplayskip}
  \setlength\abovedisplayshortskip{.3em}
  \setlength\belowdisplayshortskip{\abovedisplayskip}
}

\title{\LARGE \bf
Incremental Dissipativity based Control of Discrete-Time\\ Nonlinear Systems via the LPV Framework
}

\author{Patrick J.W. Koelewijn$^{1}$, Roland T\'oth$^{1,2}$ and Siep Weiland$^1$
\thanks{This work has received funding from the European Research Council (ERC) under the European Union’s Horizon 2020 research and innovation programme (grant agreement nr. 714663).}
\thanks{$^{1}$Patrick J.W. Koelewijn, Roland T\'oth and Siep Weiland are with Control Systems Group, Faculty of  Electrical Engineering, Eindhoven University of Technology, 5600 MB Eindhoven, The Netherlands
        {\tt\small $\lbrace$p.j.w.koelewijn, r.toth, s.weiland$\rbrace$@tue.nl}}%
\thanks{$^{2}$Roland T\'oth is also affiliated with the Systems and Control Laboratory, Institute for Computer Science and Control, Kende u. 13-17, H-1111 Budapest, Hungary
        {\tt\small toth.roland@sztaki.hu}}%
}

\begin{document}

\maketitle
\thispagestyle{empty}
\pagestyle{empty}

\begin{abstract}
Unlike for Linear Time-Invariant (LTI) systems, for nonlinear systems, there exists no general framework for systematic convex controller design which incorporates performance shaping. The Linear Parameter-Varying (LPV) framework sought to bridge this gap by extending convex LTI synthesis results such that they could be applied to nonlinear systems. However, recent literature has shown that naive application of the LPV framework can fail to guarantee the desired asymptotic stability guarantees for nonlinear systems. Incremental dissipativity theory has been successfully used in the literature to overcome these issues for Continuous-Time (CT) systems. However, so far no solution has been proposed for output-feedback based incremental control for the Discrete-Time (DT) case. Using recent results on convex analysis of incremental dissipativity for DT nonlinear systems, in this paper, we propose a convex output-feedback controller synthesis method to ensure closed-loop incremental dissipativity of DT nonlinear systems via the LPV framework. The proposed method is applied on a simulation example, demonstrating improved stability and performance properties compared to a standard LPV controller design.
\end{abstract}

\section{Introduction}
Control of nonlinear systems has been in focus of intensive research over the last few decades. However, unlike for the control of \emph{Linear Time-Invariant} (LTI) systems, there exists no systematic framework for general nonlinear systems, which incorporates performance shaping into the synthesis procedure. The \emph{Linear Parameter-Varying} (LPV) framework \cite{Shamma1988,Toth2010} aimed at bridging this gap by extending the systematic LTI controller synthesis tools to LPV systems. By embedding the nonlinear system in an LPV representation, the convex synthesis tools of the LPV framework could be used to efficiently synthesize nonlinear controllers and systematically incorporate performance shaping \cite{Wu2001}. 
However, recent research has shown that naively applying the tools of the LPV framework to nonlinear systems may fail to provide the desired performance and stability guarantees \cite{Koelewijn2020,Scorletti2015}. Namely, using the LVP framework, only asymptotic stability of the origin of the nonlinear system can be guaranteed \cite{Koelewijn2020}. However, for reference tracking and disturbance rejection, convergence to a desired steady-state trajectory is required, which the standard LPV tools cannot always guarantee. 

Hence, the use of equilibrium independent stability and performance notions are required for systematic control synthesis and analysis for nonlinear systems. Contraction \cite{Lohmiller1998}, incremental stability \cite{Angeli2002} and convergence \cite{Pavlov2006} are such equilibrium independent stability notions. Extensions of dissipativity \cite{Willems1972}, which allows for the simultaneous characterization of stability and performance, have also been made, resulting in incremental dissipativity \cite{Verhoek2020} and differential dissipativity \cite{Forni2013a} notions. In this setting, an appropriate metric of storage, e.g. energy, between trajectories or the variation of storage along the trajectories is analyzed, opposed to the standard dissipativity notions where storage, e.g. energy, is only considered with respect to a single point of neutral storage. Hence, incremental and differential dissipativity allows for equilibrium independent analysis of stability and performance whereas standard dissipativity does not. The incremental and differential stability and dissipativity results have also been extended to convex controller synthesis procedures for \emph{Continuous-Time} (CT) nonlinear systems, see \cite{Manchester2018,Scorletti2015,Koelewijn2019a,Wang2020}. Due to the equilibrium independent stability properties these synthesis results truly allow for the systematic analysis and controller design for nonlinear systems, where the LPV framework is used as a tool in order to convexify the corresponding optimization problems. 

As mentioned, most of these works have focused on analysis and control algorithms for CT nonlinear systems. For \emph{Discrete-Time} (DT) nonlinear systems, the literature is less extensive, although DT systems are widely used in system identification, model predictive control and embedded control. Even though incremental and differential stability and dissipativity results have been extended to DT systems, see \cite{Tran2018,Koelewijn2021a}, and results for state-feedback design \cite{Wei2021} and model predictive control \cite{Koehler2020} have been derived, comprehensive results for DT optimal-gain output-feedback synthesis remains an open problem. Hence, in this paper our main contribution is to develop an incremental dissipativity based convex output-feedback controller synthesis method for DT nonlinear systems, making use of the LPV framework, based on an extension of the CT results in \cite{Koelewijn2020b}. 

The paper is structured as follows. In Section \ref{sec:probdef}, a formal problem definition of the controller synthesis problem is given. In Section \ref{sec:synthesis}, the proposed controller synthesis method is described. Section \ref{sec:example} demonstrates the improved stability and performance properties of the proposed controller design compared to standard LPV control design through a simulation example. Lastly, in Section \ref{sec:conclusion}, conclusions and future recommendations are given.

\subsection{Notation}\vspace{-.5em}
The set of natural numbers including zero is denoted by $\mathbb{N}$. The set of real numbers is denoted by $\mathbb{R}$. 
The space of square-summable real valued sequences $\mathbb{N}\rightarrow\mathbb{R}$ is denoted by $\ell_2$, with the norm $\norm{x}_2=\sqrt{\sum_{k=0}^\infty\norm{x_{k}}^2}$, where $\norm{\cdot}$ denotes the Euclidian (vector) norm. A function $f$ is of class $\m{C}_n$, i.e. $f\in\m{C}_n$, if it is $n$-times continuously differentiable. A function $\alpha: \mathbb{R}_+ \rightarrow \mathbb{R}_+$ is of class $\mathcal{K}$, if it is continuous (class $\mathcal{C}$) and strictly increasing with $\alpha(0)=0$. A function $\beta: \mathbb{R}_+ \times \mathbb{N} \rightarrow \mathbb{R}_+$ is of class $\mathcal{KL}$, if for any fixed $k \geq 0$, $\beta(\cdot,k) \in \mathcal{K}$ and for any fixed $s\geq 0$, $\beta(s,\cdot)$ is decreasing with $\lim_{k\rightarrow \infty}\beta(s,k)= 0$. The column vector $\begin{bmatrix}x_1^\top &\cdots &x_n^\top\end{bmatrix}^\top$ is denoted as $\col(x_1,\dots,x_n)$. The notation $A\succ 0$ ($A\succeq 0$) indicates that $A$ is positive (semi-)definite while $A\prec 0$ ($A\preceq 0$) means that $A$ is negative (semi-)definite. The set of positive definite matrices is denoted by $\mathbb{S}_+$. 
The term that makes a matrix expression symmetric is denoted by $(\star)$, e.g. $AX+(\star)=AX+X^\top A^\top$. 
The discrete time-shift operator given by $q$ is such that for a sequence $x$, $q x_k = x_{k+1}$.

\section{Problem Definition}\label{sec:probdef}
Consider a DT nonlinear system of the form
\begin{subequations}\label{eq:nl}
\begin{align}
	x_{k+1} &= f(x_{k},w_{k});\\
	z_{k} &= h(x_{k},w_{k});
\end{align}
\end{subequations}
where $x_{k} \in \m{X} \subseteq \mathbb{R}^{n_\mr{x}}$ is the state with initial condition $x_{k=0}=x_0\in\m{X}$, $w_{k} \in \m{W} \subseteq \mathbb{R}^{n_\mr{w}}$ is the generalized disturbance, $z_{k} \in \m{Z} \subseteq \mathbb{R}^{n_\mr{z}}$ the generalized performance and $k \in \mathbb{N}$ is the discrete-time instant. The sets $\m{X}$, $\m{W}$ and $\m{Z}$ are open and convex, containing the origin. The solutions of \eqref{eq:nl} satisfy \eqref{eq:nl} in the ordinary sense and are restricted to $k\in\mathbb{N}$. The functions $f:\m{X}\times\m{W}\rightarrow\m{X}$ and $h:\m{X}\times \m{W}\rightarrow\m{Z}$ are assumed to be in $\m{C}_1$, i.e. $f,h\in\m{C}_1$, are assumed to be Lipschitz continuous and such that for all $x_0\in\m{X}$ and $w\in\m{W}^\mathbb{N}$ there is a unique solution $(x,z)\in(\m{X}\times\m{Z})^\mathbb{N}$. We define the set of solutions of \eqref{eq:nl} as
\begin{multline}
	\mathscr{B}:=\Big\lbrace (x,w,z)\in(\m{X}\times\m{W}\times\m{Z})^\mathbb{N}\mid \\(x,w,z) \text{ satisfies \eqref{eq:nl}} \Big\rbrace.
\end{multline}
We define $\mathscr{B}_\mr{x,w} := \pi_\mr{x,w}\mathscr{B}$, where $\pi_\mr{x,w}$ denotes the projection $(x,w) = \pi_\mr{x,w}(x,w,z)$.

In this paper we propose a convex output-feedback controller synthesis method to ensure incremental stability and performance of the closed-loop system of the form \eqref{eq:nl}. 
\begin{defn}[Incremental stability \cite{Tran2018}]\label{def:incrstab}
	A nonlinear systems of the form \eqref{eq:nl} is said to be incrementally asymptotically stable if there exists a function $\beta\in\m{KL}$ such that
	\begin{equation}
		\Vert x_{k}-x_k^* \Vert \leq \beta(\Vert x_0-x_0^*\Vert,k),
	\end{equation}
	for all $(x,\bar w),(x^*,\bar w)\in\mathscr{B}_\mr{x,w}$, $\bar w\in\m{W}^\mathbb{N}$ and $k\in\mathbb{N}$.
\end{defn}
While several incremental performance notions exist, such as incremental passivity, generalized incremental \htwo, etc. see \cite{Verhoek2020,Koelewijn2021a}, in this paper we focus on the performance notion in terms of the incremental \dltwo-gain. However, the theory that we develop in this paper is generally applicable to DT incremental (Q,S,R) dissipativity.
\begin{defn}[\dlitwo-gain \cite{Koelewijn2021a}]
A nonlinear system of the form \eqref{eq:nl} is said to have a finite incremental \dltwo-gain, denoted as \dlitwo-gain, if for all $w,w^*\in\ell_2$ and $x_0, x^*_0\in\m{X}$, with $(x,w,z),({x}^*,w^*,z^*) \in\mathscr{B}$
, there is a finite $\gamma\geq0$ and a function $\zeta(x,x^*)\geq0$ with $\zeta(x,x)=0$ such that
\begin{equation}\label{eq:li2gain}
	\norm{z-z^*}_2\leq\gamma\norm{w-w^*}_2+\zeta(x_0,x^*_0).
\end{equation}
The induced \dlitwo-gain of the system is the infimum of $\gamma$ such that \eqref{eq:li2gain} still holds.
\end{defn} 

Using results on incremental dissipativity theory for DT nonlinear systems from \cite{Koelewijn2021a}, the following test can be performance to analyze the induced $\dlitwo$-gain of a nonlinear system of the form \eqref{eq:nl}.

\begin{thm}[\dlitwo-gain analysis \cite{Koelewijn2021a}]\label{thm:li2gain}
	A nonlinear system of the form \eqref{eq:nl} has a finite \dlitwo-gain bounded by $\gamma$, if there exists a ${P}\succ 0$ such that for all $\mathsfit{x}\in{\m{X}}$ and $\mathsfit{w}\in{\m{W}}$
	\begin{equation}\label{eq:matli2}
		\begin{bmatrix}
			{P}		& \m{A}_\delta (\mathsfit{x},\mathsfit{w}){P} 	& \m{B}_\delta(\mathsfit{x},\mathsfit{w}) & 0\\
			\star 	& {P}						&0						& {P}\m{C}_\delta^\top(\mathsfit{x},\mathsfit{w})\\
			\star 	& \star  					&\gamma I 				& \m{D}_\delta^\top(\mathsfit{x},\mathsfit{w})\\
			\star 	& \star 						& \star 					&\gamma I
		\end{bmatrix}\succ 0,
	\end{equation} 
	where $
  \m{A}_\delta = \Partial{f}{x}$, $\m{B}_\delta = \Partial{f}{w}$, $\m{C}_\delta = \Partial{h}{x}$ and $\m{D}_\delta = \Partial{h}{w}$.
\end{thm}
The proof can be found in \cite{Koelewijn2021a}. 
Note that a bounded \dlitwo-gain implies incrementally asymptotically stability, see \cite{Koelewijn2021a}, which we will refer to as the system being \dlitwo-gain stable\footnote{Similarly, if a system has a bounded \dltwo-gain and is asymptotically stable we say it is \dltwo-gain stable.}. These results are related to the so-called dissipativity and stability of the differential form, which represents the dynamics of the variation along the trajectories of the system.
\begin{defn}[Differential form]\label{def:diffform}
	The differential form of a nonlinear system \eqref{eq:nl} is given by
	\begin{subequations}\label{eq:difform}
		\begin{alignat}{2}
			\delta x_{k+1} &= \m{A}_\delta(x_k,w_k)\delta x_k&&+\m{B}_\delta(x_k,w_k)\delta w_k;\\
			\delta z_k   &= \m{C}_\delta(x_k,w_k)\delta x_k&&+\m{D}_\delta(x_k,w_k)\delta w_k;
		\end{alignat}
	\end{subequations}
	where $\delta x_\mr{k}\in\mathbb{R}^{n_\mr{x}}$, $\delta w_{k} \in \mathbb{R}^{n_\mr{w}}$ and $\delta z_{k} \in \mathbb{R}^{n_\mr{z}}$ are the state, generalized disturbance and generalized performance associated with the differential form, respectively, and $(x,w)\in\mathscr{B}_\mr{x,w}$. In literature, \eqref{eq:difform} is also referred to as the variational or differential dynamics \cite{Crouch1987,Manchester2018}.
\end{defn}
In view of Definition \ref{def:diffform}, we are going to call \eqref{eq:nl} the primal form of \eqref{eq:difform}. The differential variables can be linked to the incremental analysis as follows. Consider $(x,w,z),({x}^*,w^*,z^*) \in\mathscr{B}$, based on which we can define a smoothly parameterized family of trajectories $(\bar x(\lambda),\bar w(\lambda),\bar z(\lambda))\in\mathfrak{B}$ with $\lambda\in[0,1]$ such that $(\bar x(1),\bar w(1),\bar z(1))=(x,w,z)$ and $(\bar x(0),\bar w(0),\bar z(0))=(x^*,w^*,z^*)$. Note that the geodesic, i.e. the minimum energy path under a given metric, corresponds to $\bar{x}(\lambda)$. The differential variables for \eqref{eq:difform} are then defined as $\delta x_k \!=\!\left.\frac{\partial\bar{x}_k(\lambda)}{\partial \lambda}\right\vert_{\lambda=1}$, $\delta w_k \!=\!\left.\frac{\partial\bar{w}_k(\lambda)}{\partial \lambda}\right\vert_{\lambda=1}$ and $\delta z_k\! =\! \left.\frac{\partial\bar{z}_k(\lambda)}{\partial \lambda}\right\vert_{\lambda=1}$, see \cite{Koelewijn2021a,Verhoek2020} for more details.

\begin{rem}\label{rem:diffprimimply}
As a result of Theorem \ref{thm:li2gain}, if the differential form \eqref{eq:difform} is \dltwo-gain stable with its \dltwo-gain bounded by $\gamma$, then the primal form \eqref{eq:nl} is \dlitwo-gain stable with a \dlitwo-gain bound of $\gamma$. More generally, standard dissipativity (with quadratic storage and supply function) of the differential form (also referred to as differential dissipativity) implies incremental dissipativity of the primal form. Importantly, by embedding the differential form in an LPV representation, inequality \eqref{eq:matli2} can be efficiently solved as a convex test using the LPV framework, see \cite{Koelewijn2021a} for more details.
\end{rem}

Similar to the systematic LTI and LPV controller synthesis frameworks \cite{Wu2001}, widely used in industry, we make use of the generalized plant concept to shape the performance of the closed-loop dynamics using LTI weighting filters. In this paper we assume that the generalized plant (with LTI weighting filters included) is of the form
\begin{subequations}\label{eq:genplant}
	\begin{align}
		x_{k+1} &= f(x_{k})+B_\mr{w} w_{k}+B_\mr{u}u_{k};\\
		z_{k} &= h_\mr{z}(x_{k})+D_\mr{zw} w_{k}+D_\mr{zu} u_{k};\\
		y_{k} &= C_\mr{y}x_k+D_\mr{yw} w_{k};
	\end{align}
\end{subequations}
where $x_k\in\m{X}\subseteq\mathbb{R}^{n_\mr{x}}$, $w_{k} \in \m{W} \subseteq \mathbb{R}^{n_\mr{w}}$ and $z_{k} \in \m{Z} \subseteq \mathbb{R}^{n_\mr{z}}$ are the state, generalized disturbance and generalized performance signals of the plant, respectively, and where $u_{k}\in\m{U}\subseteq\mathbb{R}^{n_\mr{u}}$ is the control input and $y_{k}\in\m{Y}\subseteq\mathbb{R}^{n_\mr{y}}$ is the measured output. The sets $\m{X}$, $\m{W}$, $\m{U}$, $\m{Z}$ and $\m{Y}$ are open and convex, containing the origin. The solutions of \eqref{eq:genplant} satisfy \eqref{eq:genplant} in the ordinary sense and are restricted to $k\in\mathbb{N}$. 
The functions $f:\m{X}\rightarrow\m{X}$ and $h_\mr{z}:\m{X}\rightarrow\m{Z}$ are assumed to be in $\m{C}_1$, i.e. $f,h_\mr{z}\in\m{C}_1$. Furthermore, $B_\mr{w}\in\mathbb{R}^{n_\mr{x}\times n_\mr{w}}$, $B_\mr{u}\in\mathbb{R}^{n_\mr{x}\times n_\mr{u}}$, $D_\mr{zw}\in\mathbb{R}^{n_\mr{z}\times n_\mr{w}}$, $D_\mr{zu}\in\mathbb{R}^{n_\mr{z}\times n_\mr{u}}$ and $D_\mr{yw}\in\mathbb{R}^{n_\mr{y}\times n_\mr{w}}$. 
 The solution set of \eqref{eq:genplant} is defined as\begin{multline}
	\mathfrak{B}:=\Big\lbrace (x,w,u,z,y)\in(\m{X}\times\m{W}\times\m{U}\times\m{Z}\times\m{Y})^\mathbb{N}\mid \\(x,w,u,z,y) \text{ satisfies \eqref{eq:genplant}} \Big\rbrace. 
\end{multline}
\begin{rem}
Note that while \eqref{eq:genplant} might seem restrictive, a larger class of nonlinear generalized plants of the form
\begin{subequations}\label{eq:genplantbig}
	\begin{align}
		x_{k+1} &= f(x_{k},u_k)+B_\mr{w} w_{k};\\
		z_{k} &= h_\mr{z}(x_{k},u_k)+D_\mr{zw} w_{k} ;\\
		y_{k} &= h_\mr{y}(x_{k})+ D_\mr{yw}w_k;
	\end{align}
\end{subequations}
where $f:\m{X}\times\m{U}\rightarrow\m{X}$, $h_\mr{z}:\m{X}\times\m{U}\rightarrow\m{Z}$ and $h_\mr{y}:\m{X}\rightarrow\m{Y}$ with $f,h_\mr{z},h_\mr{y}\in\m{C}_1$,
	can be written in the form \eqref{eq:genplant} by interconnecting appropriate filters to $u$ and $y$ of \eqref{eq:genplantbig}.
\end{rem}

 The to-be-designed controller for the generalized plant \eqref{eq:genplant} is of the form\begin{subequations}\label{eq:controller}
	\begin{align}
		x_{\mr{c},k+1} &= f_\mr{c}(x_{\mr{c},k},u_{\mr{c},k});\\
		y_{\mr{c},k} &= h_\mr{c}(x_{\mr{c},k},u_{\mr{c},k});
	\end{align}
\end{subequations}
where $x_\mr{c}$ is the state, $u_\mr{c}$ is the input and $y_\mr{c}$ is the output of the controller. The closed-loop interconnection of a generalized plant $P$ given by \eqref{eq:genplant} and a controller $K$ given by \eqref{eq:controller} with $u_\mr{c} = y$ and $u = y_\mr{c}$ is denoted by $\m{F}_l(P,K)$, which is assumed to be well-posed and hence in the form \eqref{eq:nl}. In this paper we propose a convex controller synthesis method such that the closed-loop interconnection $\m{F}_l(P,K)$ is \dlitwo-gain stable with minimal \dlitwo-gain.

\section{Controller Synthesis Method}\label{sec:synthesis}\vspace{-.5em}
\subsection{Overview}\vspace{-.5em}
In this section, the proposed controller synthesis method is discussed. In order to obtain a DT controller ensuring closed-loop \dlitwo-gain stability and performance the following procedure is proposed, which follows along the same lines as the CT version in \cite{Koelewijn2020b}:
\begin{enumerate}
	\item For the generalized plant \eqref{eq:genplant}, its differential form is computed, which is then embedded in an LPV representation.
	\item For the LPV embedding of the differential form of the generalized plant, an LPV controller is synthesized such that the closed-loop interconnection is \dltwo-gain stable with minimal \dltwo-gain $\gamma$. This controller will be referred to as the differential (LPV) controller. \label{itm:step2}
	\item The differential controller designed in Step \ref{itm:step2} is realized into a primal form. The resulting closed-loop interconnection of the primal form of the generalized plant and realized primal form of the controller then is \dlitwo-gain stable with a \dlitwo-gain bounded by $\gamma$.
\end{enumerate}
As contribution of this paper, we show how these steps can be accomplished and we provide proofs of the mentioned implications.
\subsection{LPV embedding of the generalized plant}\label{sec:stepA}\vspace{-.5em}
As a first step in our proposed controller synthesis procedure, the differential form of the generalized plant is computed and embedded in an LPV representation. The differential form of \eqref{eq:genplant} is given by
\begin{subequations}\label{eq:difgenplant}
	\begin{align}
		\delta x_{k+1} &= A_\delta(x_k)\delta x_k +B_\mr{w} \delta w_k+B_\mr{u}\delta u_k;\\
		\delta z_k &= C_{\delta\mr{z}}(x_k) \delta x_k+D_\mr{zw} \delta w_k+D_\mr{ zu} \delta u_k;\\
		\delta y_k &= C_{\mr{y}} \delta x_k +D_\mr{yw} \delta w_k;\label{eq:difgenplanty}
	\end{align}
\end{subequations}
where $A_\delta = \Partial{f}{x}$, $C_{\delta\mr{z}} = \Partial{h_\mr{z}}{x}$ and where $x\in\pi_\mr{x}\mathfrak{B}$ is the state of the primal form \eqref{eq:genplant}. The differential form of the generalized plant \eqref{eq:difgenplant} is then embedded in an LPV representation, resulting in an LPV model \eqref{eq:difgenplantlpv}, in accordance with the following definition.
 \begin{defn}[LPV embedding]\label{def:lpvemb}
	Assume we have a nonlinear system of the form \eqref{eq:genplant} with differential form given by \eqref{eq:difgenplant}. The LPV state-space model 
\begin{subequations}\label{eq:difgenplantlpv}
	\begin{align}
		\delta x_{k+1} &= A(\rho_k)\delta x_k +B_\mr{w} \delta w_k+B_\mr{u}\delta u_k;\\
		\delta z_k &= C_{\mr{z}}(\rho_k) \delta x_k+D_\mr{zw} \delta w_k+D_\mr{zu}\delta u_k;\\
		\delta y_k &= C_{\mr{y}} \delta x_k +D_\mr{yw} \delta w_k;
	\end{align}
\end{subequations}
	where $\rho_k\in\m{P}\subset\mathbb{R}^{n_\rho}$ is the scheduling-variable, is an LPV embedding of the differential form \eqref{eq:difgenplant} on the region $\mathscr{X}\supseteq\m{X}$, if there exists a function $\psi:\mathbb{R}^{n_\mr{x}}\rightarrow\mathbb{R}^{n_\rho}$, called the scheduling-map, such that under a given choice of function class $\mathfrak{A}$ for ${A},C_\mr{z}$, e.g. affine, polynomial, etc., $A(\psi(x))=A_\delta(x)$, ${C}_\mr{z}(\psi(x))={C}_\mr{\delta z}(x)$ for all $x\in\mathscr{X}$ and $\psi(\mathscr{X})\subseteq \m{P}$ where $\m{P}$ is a (minimal) convex hull with $n$ vertices. 
	\end{defn}

Next, we use the LPV embedding of the differential form of the generalized plant \eqref{eq:difgenplantlpv} in order to be able to use convex controller synthesis.
\vspace{-.2em}
\subsection{Differential controller synthesis}\label{sec:diffsyn}\vspace{-.3em}
In this step, a controller for the differential form of the generalized plant \eqref{eq:difgenplant} is synthesized such that the closed-loop interconnection is \dltwo-gain stable with minimal \dltwo-gain. To convexify this problem, the LPV framework is used to perform this step. Hence, we synthesize an LPV controller for the LPV embedding of the differential form of the generalized plant \eqref{eq:difgenplantlpv}, obtained in the previous step (Section \ref{sec:stepA}). 
The LPV controller is assumed to be of the form\begin{subequations}\label{eq:difcontroller}
\begin{align}
	\delta x_{\mr{c},k+1} &= A_{\delta\mr{c}}(\rho_k)\delta x_{\mr{c},k}+B_{\delta\mr{c}}(\rho_k)\delta u_{\mr{c},k};\label{eq:difcontrollerx}\\
	\delta y_{\mr{c},k} &= C_{\delta\mr{c}}(\rho_k)\delta x_{\mr{c},k}+D_{\delta\mr{c}}(\rho_k)\delta u_{\mr{c},k};\label{eq:difcontrollery}
\end{align}
\end{subequations}
where $x_{\mr{c},k}\in\mathbb{R}^{n_\mr{x_c}}$ is the state, $u_{\mr{c},k}\in\mathbb{R}^{n_\mr{y}}$ is the input and $y_{\mr{c},k}\in\mathbb{R}^{n_\mr{u}}$ is the output of the controller and $A_{\delta\mr{c}},\dots,D_{\delta\mr{c}}\in\mathfrak{A}$. We will refer to \eqref{eq:difcontroller} as the differential controller.
 Various methods exists to obtain an LPV controller \eqref{eq:difcontroller}, i.e. to synthesize a DT LPV controller minimizing the \dltwo-gain of the closed-loop system. Next, we will briefly describe one particular method.

 \begin{lem}[Differential LPV controller synthesis]\label{lem:lpvsyn}
 There exists a controller \eqref{eq:difcontroller} such that the closed-loop interconnection of \eqref{eq:difgenplantlpv} and \eqref{eq:difcontroller} is \dltwo-gain stable and has an \dltwo-gain bounded by $\gamma$ if 
 there exists matrices $\m{P}_\mr{x}$, $\m{P}_\mr{z}\in\mathbb{S}^{n_\mr{x}}_+$, matrices $\m{P}_\mr{y},J,N,S\in\mathbb{R}^{n_\mr{x}\times n_\mr{x}}$ and matrix functions $U(\rho)\in\mathbb{R}^{n_\mr{x}\times n_\mr{x}}$, $V(\rho)\in\mathbb{R}^{n_\mr{x}\times n_\mr{y}}$, $W(\rho)\in\mathbb{R}^{n_\mr{u}\times n_\mr{x}}$, $X(\rho)\in\mathbb{R}^{n_\mr{u}\times n_\mr{y}}$ s.t.
	\begin{equation}\label{eq:synlmi}
		\begin{bmatrix}
			\m{P} & \m{A}(\rho) & \m{B}(\rho) & 0\\
			\star & \m{G} & 0& \m{C}(\rho)^\top\\
			\star & \star & \gamma I & \m{D}(\rho)^\top\\
			\star & \star & \star & \gamma I
		\end{bmatrix}\succ 0,
	\end{equation}
	where
\begin{gather}
	\m{P} = \begin{bmatrix}
		\m{P}_\mr{x} & \m{P}_\mr{y}\\\m{P}_\mr{y}^\top  &\m{P}_\mr{z}
	\end{bmatrix}, \quad G = \begin{bmatrix}
		J+J^\top-\m{P}_\mr{x}&I+S^\top-\m{P}_\mr{y}\\\star & N+N^\top-\m{P}_\mr{z}
	\end{bmatrix},\notag\\
	\m{A}(\rho) = \begin{bmatrix}
		A(\rho)J+B_\mr{u}W(\rho)& A+B_\mr{u}X(\rho)C_\mr{y}\\
		U(\rho)&NA(\rho)+V(\rho)C_\mr{y},
			\end{bmatrix},\notag \\
	\m{B}(\rho) = \begin{bmatrix}
		B_\mr{w}+B_\mr{u}X(\rho)D_\mr{yw}\\
		NB_\mr{w}+V(\rho)D_\mr{yw}
			\end{bmatrix},\\ \m{C}(\rho) = \begin{bmatrix}
				C_\mr{z}(\rho)J+D_\mr{zu}W(\rho) & C_\mr{z}(\rho)+D_\mr{zu}X(\rho)C_\mr{y}
			\end{bmatrix},\notag\\
	\m{D}(\rho) = D_\mr{zw}+D_\mr{zu}X(\rho)D_\mr{yw}.\notag
\end{gather}
\end{lem}
See \cite{Ali2011} for the proof of Lemma \ref{lem:lpvsyn}.
Note that $\gamma$ appears linearly in \eqref{eq:synlmi}, hence, it can be minimized when solving the LMI \eqref{eq:synlmi}.
If a solution to \eqref{eq:synlmi} has been found, the matrices for \eqref{eq:difcontroller} can be constructed by first finding matrices $R, L\in\mathbb{R}^{n_\mr{x}\times n_\mr{x}}$ such that $S=NJ+RL$, and then computing
\begin{equation}\label{eq:contrconstruct}
\hspace{-.1em}\begin{bmatrix}
	A_\mr{\delta c}(\rho) & B_\mr{\delta c}(\rho)\\C_\mr{\delta c}(\rho) & D_\mr{\delta c}(\rho)
\end{bmatrix}\!=\!\begin{bmatrix}
	R & N B_\mr{u}\\0 & I
\end{bmatrix}^{-1}\!\Theta(\rho)\!\begin{bmatrix}
	L & 0\\C_\mr{y}J&I
\end{bmatrix}^{-1}\hspace{-.5em},\hspace{-.4em}
\end{equation}
where
\begin{equation}\label{eq:contconstruct2}
		\Theta(\rho) = \begin{bmatrix}
	U(\rho)&V(\rho)\\W(\rho)&X(\rho)
\end{bmatrix}-\begin{bmatrix}
	NA(\rho)J & 0\\0&0
\end{bmatrix}.
\end{equation}
If \eqref{eq:difgenplantlpv}, $U(\rho)$, $V(\rho)$, $W(\rho)$ and $X(\rho)$ have an affine dependency on the scheduling-variable, then \eqref{eq:synlmi} can be solved as a convex optimization problem using LMIs, see \cite{Ali2011}, and \eqref{eq:difcontroller} can be recovered with affine dependency in terms of \eqref{eq:contrconstruct}, \eqref{eq:contconstruct2}. Note the method described above uses a (quadratic) scheduling-independent storage function. 

\begin{thm}[Differential closed-loop \dltwo-gain]\label{thm:diffclpl2}
	The closed-loop interconnection of the differential form of the generalized plant $\delta P$ given by \eqref{eq:difgenplant} and the controller $\delta K$ given by \eqref{eq:difcontroller}, denoted by $\m{F}_l(\delta P, \delta K)$, is \dltwo-gain stable and has an \dltwo-gain bounded by $\gamma$ for all $x\in\pi_\mr{x}\mathfrak{B}$ if the closed-loop interconnection of the LPV embedding of the differential form of the generalized plant $\delta P_\mr{LPV}$ given by \eqref{eq:difgenplantlpv} and $K$, denoted by $\m{F}_l(\delta P_\mr{LPV}, \delta K)$, is \dltwo-gain stable and with a bounded \dltwo-gain of $\gamma$ for all $\rho\in\m{P}^\mathbb{N}$. 
\end{thm}
\begin{proof}
	Through the LPV embedding $\delta P_\mr{LPV}$ \eqref{eq:difgenplantlpv}, we have that $\rho=\psi(x)$ and $\psi(\mathscr{X})\subseteq\m{P}$, hence, if $x\in\mathscr{X}^\mathbb{N}$ then $\rho\in\m{P}^\mathbb{N}$. 
	Consequently, if $\m{F}_l(\delta P_\mr{LPV}, \delta K)$ is \dltwo-gain stable and has an \dltwo-gain of $\gamma$ for all $\rho\in\m{P}^\mathbb{N}$, then $\m{F}_l(\delta P, \delta K)$ is \dltwo-gain stable and its \dltwo-gain is bounded by $\gamma$ for all $x\in\mathscr{X}^\mathbb{N}$, hence, this also holds for all $x\in\pi_\mr{x}\mathfrak{B}$ as $\pi_\mr{x}\mathfrak{B}\subseteq\mathscr{X}^\mathbb{N}$.	
\end{proof}\vspace{-.5em}
Next it is shown how to realize the primal form of the controller based on synthesized differential controller \eqref{eq:difcontroller} such that the closed-loop interconnection of the primal form of the controller and primal form of the generalized plant \eqref{eq:genplant} is \dlitwo-gain stable.


\subsection{Primal controller realization}\vspace{-.2em}
Inspired by the work for CT systems \cite{Manchester2018,Koelewijn2020b} and for DT state-feedback design \cite{Wei2021} we make use of a path integral based realization to obtain the primal form of the controller that enforces convergence of the plant response towards a desired steady-state response $(x^*,w^*,u^*,z^*,y^*)\in\mathfrak{B}$. Let us denote the state of $\m{F}_l(\delta P, \delta K)$ and $\m{F}_l( P, K)$ as $\delta \chi_k=\col(\delta x_k,\delta x_{\mr{c},k})\in\mathbb{R}^{n_\mr{x}+n_\mr{x_\mr{c}}}$ and $\chi_k=\col(x_k, x_{\mr{c},k})\in\mathbb{R}^{n_\mr{x}+n_\mr{x_\mr{c}}}$, respectively.
\begin{thm}[Primal controller realization]\label{thm:primreal}
Given a differential controller $\eqref{eq:difcontroller}$ synthesized for $\delta P$ \eqref{eq:difgenplant} such that the closed-loop is \dltwo-gain stable under a (differential) storage function of the form $\delta \chi^\top P\delta \chi$ where $P\succ 0$, a primal realization of $\delta K$ is given by\begin{subequations}\label{eq:realizecontroller}
	\begin{align}
	\Delta{x}_{\mr{c},k+1} &= A_{\mr{c},k} \Delta{x}_{\mr{c},k}+B_{\mr{c},k}  (u_{\mr{c},k}-u_{\mr{c},k}^*);\\
	y_{\mr{c},k} &= y_{\mr{c},k}^* +C_{\mr{c},k} \Delta{x}_{\mr{c},k}+D_{\mr{c},k}  (u_{\mr{c},k}-u_{\mr{c},k}^*);
\end{align}
\end{subequations}
where $\Delta x_{\mr{c,k}}\in\mathbb{R}^{n_\mr{x_c}}$, $(y_\mr{c}^*,u_\mr{c}^*) = (u^*,y^*)\in\pi_{\mr{u,y}}\mathfrak{B}$ is a feasible steady-state trajectory of the plant \eqref{eq:genplant}, and
\begin{equation}\label{eq:DTcontrMat}
\begin{alignedat}{3}
	{A}_{\mr{c},k}  \!&=\! \int_0^1\!A_{\delta\mr{c}}\left(\bar{\rho}_k(\lambda)\right)d\lambda, \hspace{.5em} 
	&&{B}_{\mr{c},k}  \!=\! \int_0^1\! B_{\delta\mr{c}}\left(\bar{\rho}_k(\lambda)\right)d\lambda\\
	{C}_{\mr{c},k}  \!&=\! \int_0^1 \!C_{\delta\mr{c}}\left(\bar{\rho}_k(\lambda)\right)d\lambda, \hspace{.5em} 
		&&{D}_{\mr{c},k}  \!=\! \int_0^1\! D_{\delta\mr{c}}\left(\bar{\rho}_k(\lambda)\right)d\lambda,
\end{alignedat}\hspace{-.2em}
\end{equation}
with $\bar{\rho}_k(\lambda) = \psi\left(\bar{x}_k(\lambda)\right)$ and $\bar{x}_k(\lambda)=x_k^*+\lambda(x_k-x_k^*)$.
\end{thm}
\begin{proof}
Define a smoothly parameterized family of trajectories $(\bar x_\mr{c}(\lambda),\bar u_\mr{c}(\lambda),\bar y_\mr{c}(\lambda))$ such that $(\bar x_\mr{c}(1),\bar u_\mr{c}(1),\bar y_\mr{c}(1))=(x_\mr{c},u_\mr{c},y_\mr{c})$ and $(\bar x_\mr{c}(0),\bar u_\mr{c}(0),\bar y_\mr{c}(0))=(x_\mr{c}^*,u_\mr{c}^*,y_\mr{c}^*)$. The differential closed-loop storage function is of the form $\delta \chi^\top P\delta \chi$, corresponding to a constant Riemannian metric, see \cite{Manchester2018}. Hence, the geodesic connecting $\chi_{k}$ and $\chi_{k}^*$ is given by $\bar{\chi}_{k}(\lambda) = \chi_{k}^*+\lambda(\chi_{k}-\chi_{k}^*)$. Consequently, $\frac{\partial}{\partial \lambda}\bar{x}_{k}(\lambda) = x_{k}-x_{k}^*$ and $\frac{\partial}{\partial \lambda}\bar{x}_{\mr{c},k}(\lambda) = x_{\mr{c},k}-x_{\mr{c},k}^*=\Delta x_{\mr{c},k}$. Define $\bar{w}_{k}(\lambda) := w_{k}^*+\lambda(w_{k}-w_{k}^*)$. Due to the linearity of \eqref{eq:difgenplanty} and the definitions of $\bar{x}_{k}(\lambda)$ and $\bar{w}_{k}(\lambda)$ we obtain $\bar y_{k}(\lambda) \!=\!\bar u_{\mr{c},k}(\lambda)\! =\! u_{\mr{c},k}^*+\lambda (u_{\mr{c},k}-u_{\mr{c},k}^*)$. Based on these we can show:\begin{subequations}\label{eq:contrlambda}
	\begin{align}
		\bar x_{\mr{c},k+1}(\lambda) &= x_{\mr{c},k+1}^*+\int_0^\lambda A_{\delta\mr{c}}(\bar\rho_k(\lambda))\Delta x_{\mr{c},k}\,d\lambda+\label{eq:contrxlambda}\\&\hphantom{,=x_{\mr{c},k+1}+}\int_0^\lambda B_{\delta\mr{c}}(\bar\rho_k(\lambda))(u_{\mr{c},k}-u_{\mr{c},k}^*)\,d\lambda,\notag\\
		\bar y_{\mr{c},k}(\lambda) &= y_{\mr{c},k}^*+\int_0^\lambda C_{\delta\mr{c}}(\bar\rho_k(\lambda))\Delta x_{\mr{c},k}\,d\lambda+\label{eq:contrylambda}\\&\hphantom{,=y_{\mr{c},k}^*+}\int_0^\lambda D_{\delta\mr{c}}(\bar\rho_k(\lambda))(u_{\mr{c},k}-u_{\mr{c},k}^*)\,d\lambda.\notag
	\end{align}
	\end{subequations}
	Taking $\lambda = 1$ for \eqref{eq:contrlambda} results in \eqref{eq:realizecontroller}. Furthermore, as $\delta x_{\mr{c},k}=\left.\frac{\partial}{\partial \lambda}\bar{x}_{\mr{c},k}(\lambda)\right\vert_{\lambda=1}$, $\delta y_{\mr{c},k}=\left.\frac{\partial}{\partial \lambda}\bar{y}_{\mr{c},k}(\lambda)\right\vert_{\lambda=1}$ and $\delta u_{\mr{c},k}=\left.\frac{\partial}{\partial \lambda}\bar{u}_{\mr{c},k}(\lambda)\right\vert_{\lambda=1}$, taking the derivative of \eqref{eq:contrlambda} w.r.t. $\lambda$ and taking $\lambda=1$ results in \eqref{eq:difcontroller}. See the proof of \cite[Theorem 22]{Koelewijn2020b} for the full derivation to obtain \eqref{eq:contrlambda}.
\end{proof}\vspace{-.5em}
We will refer to \eqref{eq:realizecontroller} as the incremental LPV controller. Note that the incremental LPV controller consists of a feedback part, to converge towards the steady-state trajectory, and a feedforward part, corresponding to the steady-state trajectory.

\begin{thm}[Closed-loop \dlitwo-gain]
	The closed-loop interconnection of a generalized plant $P$ given by \eqref{eq:genplant} and controller given by \eqref{eq:realizecontroller} is \dlitwo-gain stable and its \dlitwo-gain is bounded by $\gamma$, i.e. satisfies \eqref{eq:li2gain}, if the closed-loop interconnection of the LPV model \eqref{eq:difgenplantlpv} and the LPV controller \eqref{eq:difcontroller} is \dltwo-gain stable and has a bounded \dltwo-gain of $\gamma$ for all $\rho\in\m{P}^\mathbb{N}$ under a (differential) quadratic parameter-independent storage function.
\end{thm}
\begin{proof}
	Theorem \ref{thm:diffclpl2} shows that the closed-loop interconnection of differential form of the generalized plant \eqref{eq:difgenplant} and LPV controller \eqref{eq:difcontroller} is \dltwo-gain stable and its \dltwo-gain is bounded by $\gamma$ for all $x\in\pi_\mr{x}\mathfrak{B}$ if the closed-loop interconnection of LPV model \eqref{eq:difgenplantlpv} and LPV controller \eqref{eq:difcontroller} is \dltwo-gain stable with a bounded \dltwo-gain of $\gamma$ for all $\rho\in\m{P}$. Theorem \ref{thm:primreal} shows that for the controller \eqref{eq:controller} its differential form is given by \eqref{eq:difcontroller}. This implies by Theorem \ref{thm:li2gain} that the closed-loop interconnection of the primal form of the plant \eqref{eq:genplant} and realized primal form of the controller is \dlitwo-gain stable and its \dlitwo-gain is bounded $\gamma$. Furthermore, as the closed-loop is incrementally asymptotically stable, meaning all trajectories converge towards each other, and as the steady-state trajectory $(x^*,w^*,u^*,z^*,y^*)\in\mathfrak{B}$ is by design of the controller \eqref{eq:controller} a feasible trajectory, all trajectories $(x,w,u,z,y)\in\mathfrak{B}$ will converge towards $(x^*,w^*,u^*,z^*,y^*)\in\mathfrak{B}$, i.e. $(x,w,u,z,y)\rightarrow(x^*,w^*,u^*,z^*,y^*)$ as $k\rightarrow\infty$, for $w_k\rightarrow w_k^*$ as $k\rightarrow\infty$.
	\end{proof}\vspace{-.5em}

Note that proposed incremental LPV controller explicitly depends on $(u^*,y^*)\in\pi_{\mr{u,y}}\mathfrak{B}$ corresponding to $(x^*,w^*,u^*,z^*,y^*)\in\mathfrak{B}$, hence, explicit knowledge of $w^*$ is required. As $w^*$ in the generalized plant framework can contain besides known disturbances, e.g. references, also unknown disturbances, a disturbance observer is required in order to estimate the unknown entries of $w^*$. Further details on this topic are out of the scope of the current paper and we refer the reader to \cite{Chen2016} for more details on disturbance observers and \cite{Koelewijn2020b} for application to an incremental controller design in the CT case.

\section{Example}\label{sec:example}
In this section we demonstrate the proposed incremental LPV controller synthesis method on a simulation example. For comparison, a standard LPV controller, ensuring \dltwo-gain stability, will also be designed. 

Consider the following DT nonlinear plant
\begin{subequations}\label{eq:1dsys}
	\begin{align}
		x_{1,k+1} &= 0.1x_{1,k}-x_{2,k};\\
		x_{2,k+1} &= 0.9\sin(x_{1,k})+x_{2,k}+u_k;\\
		y_k &= x_{1,k}.
	\end{align}
\end{subequations}
For this plant we want to design a controller which achieves reference tracking. 
The generalized plant structure that is taken in order to achieve this objective is depicted in Fig. \ref{fig:genplant_int}, where $G$ is the plant \eqref{eq:1dsys}, $K$ is the to-be-synthesized controller, $r$ is the reference, $W_\mathrm{e}(q) = \frac{0.2 (q-0.5)}{q+\alpha}$, $M(q) = \frac{q+\alpha}{q-1}$, and $W_\mathrm{u} = 0.2$, where $\alpha = \frac{1}{\pi}$. 

For the synthesis of the controller using the procedure described in Section \ref{sec:synthesis}, we require the differential form of the generalized plant to be embedded in an LPV representation. As the plant \eqref{eq:1dsys} is the only nonlinear system in the generalized plant (the weighting filters are linear), we only require computation of the differential form and the accompanying LPV embedding of \eqref{eq:1dsys} (as the dynamics of differential form of an LTI system are equivalent to its primal form). The following LPV embedding on the region $\mathscr{X}$ of the differential form of \eqref{eq:1dsys} is taken:
\begin{subequations}\label{eq:1dsysdiff}
	\begin{align}
		\delta x_{1,k+1} &= 0.1\delta x_{1,k}-\delta x_{2,k};\\
		\delta x_{2,k+1} &= 0.9\rho_k\delta x_{1,k}+\delta x_{2,k}+\delta u_k;\\
		\delta y_k &= \delta x_{1,k};
	\end{align}
\end{subequations}
where $\rho_k=\cos(x_{1,k})\in[-1,1]$ such that $\psi(x) = \cos(x_{1,k})$ with $x_k\in\mathscr{X}=\mathbb{R}^2$.
Synthesizing a controller for this system using the synthesis procedure described in Section \ref{sec:synthesis}, where the synthesis of the differential controller in Step \ref{itm:step2} of the synthesis method, see Section \ref{sec:diffsyn}, is performed using the method described Lemma \ref{lem:lpvsyn}. This synthesis procedure results in a closed-loop \dlitwo-gain bound of 1.1.

For comparison, a standard LPV controller is also synthesized in order to achieve a bounded closed-loop \dltwo-gain. Whereby the same generalized plant structure as depicted in Fig. \ref{fig:genplant_int} is taken. To perform standard LPV synthesis, the primal form of the plant \eqref{eq:1dsys} is embedded in an LPV representation on the region $\mathscr{X}_\mr{s}$, given by
\begin{subequations}\label{eq:1dsyslpv}
	\begin{align}
		x_{1,k+1} &= 0.1 x_{1,k}- x_{2,k};\\
		x_{2,k+1} &= 0.9\rho_{\mr{s},k} x_{1,k}+ x_{2,k}+ u_k;\\
		y_k &=  x_{1,k}.
	\end{align}
\end{subequations}
where $\rho_{\mr{s},k} = \sinc(x_{1,k})\in[-0.22,1]$ such that $\psi_\mr{s}(x,u) = \sinc(x_{1,k})$ with $x_k\in\mathscr{X}_\mr{s}=\mathbb{R}^2$. For synthesis of the standard LPV controller also the method described in Lemma \ref{lem:lpvsyn} is used (however, applied to \eqref{eq:1dsyslpv}), which results in a closed-loop \dltwo-gain bound of 0.80. 

The closed-loop systems with incremental LPV controller and standard LPV controller are both simulated for a reference $r_k = 1$ and $r_k = 2$. For the incremental LPV controller, this corresponds to the steady-state trajectory $x_{1,k}^*=r_k$, with $u_k^* =y_{\mr{c},k} = 0.9\sin(x_{1,k}^*)$. The trajectories of the closed-loop systems for both of these controllers can be found in the top two graphs in Fig. \ref{fig:ref1}. From the figure, it can be seen that for both references the incremental LPV controller achieves similar tracking behavior and it asymptotically converges towards the reference. However, the output of the plant ends up in a limit cycle around the reference when using the standard LPV controller for the reference $r_k=2$. The closed-loop system with standard LPV controller displays similar issues as have been observed in the CT case \cite{Koelewijn2020}. Furthermore, the incremental LPV controller also allows to track and guarantee convergence towards more complex reference trajectories. In the bottom graph in Fig. \ref{fig:ref1}, the reference $r_k=\sin(\tfrac{\pi}{8}k)+2.5$ is used. For this reference the corresponding feedforward trajectory $u_\mr{k}^*$ (which is not given due to its complexity) is also added to the output of the standard LPV controller to have a fair comparison with the incremental LPV controller, which by design uses this feedforward action. However, it can again be seen that also for this reference, the standard LPV controller is not able to guarantee convergence, even when feedforward is used.

\begin{figure}
	\centering
	\vspace{.5em}
	\includegraphics[scale=.8]{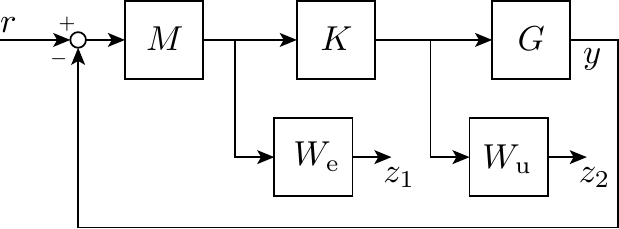}\vspace{-.5em}
	\caption{Generalized plant.}
	\label{fig:genplant_int}\vspace{-1.5em}
\end{figure}

\begin{figure}
	\centering
	\includegraphics{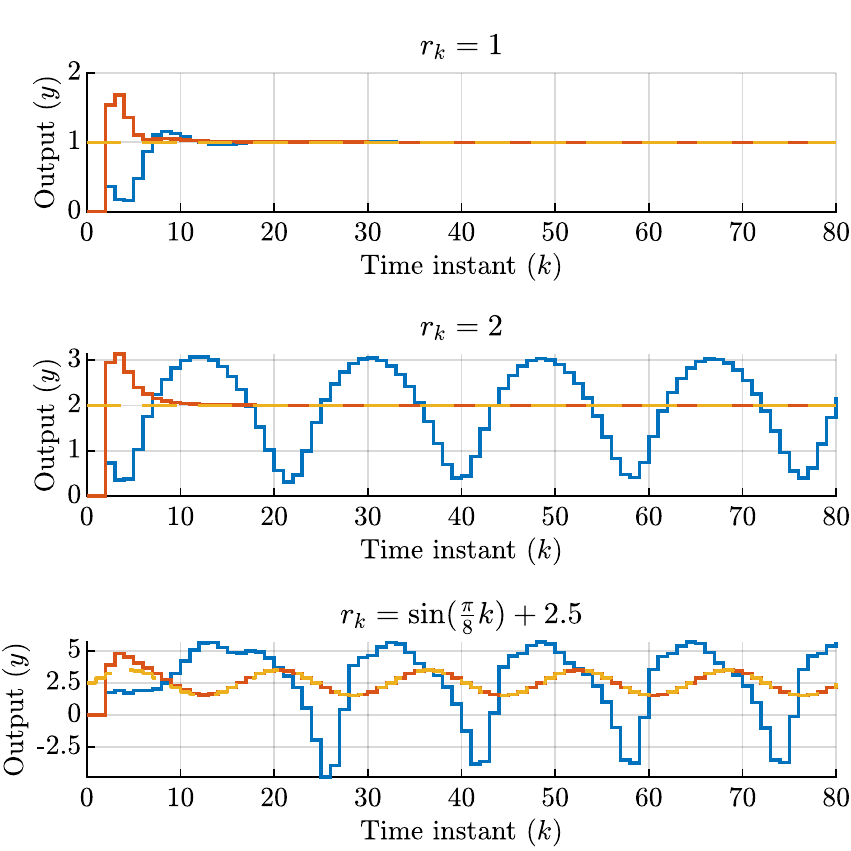}\vspace{-.5em}
	\caption{Output response of the closed-loop of the plant with standard LPV controller (\legendline{mblue}) and the incremental LPV controller (\legendline{morange}) for the reference trajectory (\legendline{myellow,dashed}).}\vspace{-2em}
	\label{fig:ref1}
\end{figure}

\section{Conclusion}\label{sec:conclusion}
In this paper, we have proposed a convex output-feedback controller synthesis method to ensure incremental dissipativity and bounded incremental \dltwo-gain for DT nonlinear systems. This is achieved by using the recent results on convex incremental dissipativity analysis of DT nonlinear systems using the LPV framework and extending the CT incremental LPV controller synthesis results. The proposed synthesis method enables systematic convex controller design for DT nonlinear systems, unlike standard LPV synthesis applied to nonlinear systems, which can have asymptotic stability issues. For future research, we aim at developing model predictive control method on the basis of incremental dissipativity theory. 



\bibliographystyle{ieeetr}
\bibliography{references.bib}

\end{document}

%% file: CustomCommands.tex
\usepackage{amsmath}
\usepackage{tikz}
\usepackage{xspace}
\usepackage{lipsum}

\DeclareRobustCommand{\legendline}[1]{\hspace{-3pt}
\tikz[#1,line width=0.4mm,baseline=-0.5ex]{\draw (0,0) -- (.35,0);}
\hspace{-3pt}}

\definecolor{mblue}{rgb}{0,0.4470,0.7410}
\definecolor{morange}{rgb}{0.8500,0.3250,0.0980}
\definecolor{myellow}{rgb}{0.9290,0.6940,0.1250}
\definecolor{mpurple}{rgb}{0.4940,0.1840,0.5560}
\definecolor{mgreen}{rgb}{0.4660,0.6740,0.1880}
\definecolor{mcyan}{rgb}{0.3010,0.7450,0.9330}
\definecolor{mred}{rgb}{0.6350,0.0780,0.1840}
\definecolor{mgreenblue}{rgb}{0.0,1.0,0.5}
\definecolor{parulablue}{rgb}{0.2431,0.1490,0.6588}
\definecolor{parulalblue}{RGB}{39,151,235}
\definecolor{parulagreen}{RGB}{129,204,89}
\definecolor{parulayellow}{RGB}{249,251,21}

\newcommand{\norm}[1]{\left\lVert#1\right\rVert}

\newcommand{\dlitwo}{\ensuremath{\ell_{\mathrm{i}2}}\xspace}
\newcommand{\dltwo}{\ensuremath{\ell_2}\xspace}
\newcommand{\m}[1]{\mathcal{#1}}
\newcommand{\mr}[1]{\mathrm{#1}}

\newcommand{\htwo}{\ensuremath{\mathcal{H}_2}\xspace}
\newcommand{\Partial}[3][]{\frac{\partial^{#1} #2}{\partial #3^{#1}}}

\DeclareMathOperator{\sinc}{sinc}
\DeclareMathOperator{\col}{col}